\documentclass[12pt]{amsart}
\usepackage{amssymb,latexsym,amsmath,amsthm,enumitem,mathrsfs,geometry}
\usepackage{fullpage}
\usepackage{fancyhdr}
\usepackage{xcolor}
\usepackage{bbm}
\usepackage{stmaryrd}
\usepackage{mathrsfs}
\usepackage{hyperref}
\usepackage{lineno}
\usepackage{diagbox}
\usepackage{array}
\usepackage{tikz}
\usetikzlibrary{arrows}
\usetikzlibrary{positioning}
\usepackage{float}
\usepackage{comment}
\usepackage{mathtools}
\usepackage{cleveref}
\usepackage{graphicx}
\usepackage{subcaption}

\newtheorem{theorem}{Theorem}[section]
\newtheorem*{theorem*}{Theorem}
\newtheorem{lemma}{Lemma}[section]
\newtheorem*{remark*}{Remark}

\begin{document}

\numberwithin{equation}{section}

\title{Lie symmetry classification and invariant solutions of time-fractional telegraph systems with variable coefficients}

\author[Sodoo]{Sodbaatar Adiya}
\address{National Intelligence University and Department of Mathematics, National University of Mongolia, Ulaanbaatar, Mongolia}
\email{sodoomath@gmail.com}
\author[khongoroo]{Khongorzul Dorjgotov}
\address{Department of Mathematics, National University of Mongolia, Ulaanbaatar, Mongolia}
\email{pilpalpil@gmail.com}
\author[magnai]{Bayarmagnai Gombodorj}
\address{Department of Mathematics, National University of Mongolia, Ulaanbaatar, Mongolia}
\email{bayarmagnai@smcs.num.edu.mn}
\author[bayaraa]{Bayarpurev Mongol}
\address{Department of Electronics and Communication Engineering, National University of Mongolia, Ulaanbaatar, Mongolia}
\email{bayarpurev@gmail.com}
\author[uuganaa]{Uuganbayar Zunderiya$^*$}
\address{Department of Mathematics, National University of Mongolia, Ulaanbaatar, Mongolia}
\email{zunderiya@gmail.com}

\thanks{$^*$Corresponding Author: U. Zunderiya}


\begin{abstract}
Time-fractional telegraph equations provide fundamental mathematical models for transport processes that exhibit memory and nonlocal effects in industrial and physical systems. These models arise naturally in heat transport in materials with thermal memory, wave propagation in viscoelastic media, and charge transport in spatially heterogeneous semiconductor devices. In this study, we investigate a class of time-fractional telegraph systems with spatially varying coefficients using Lie symmetry analysis and the Riemann--Liouville fractional derivative. We establish a complete Lie group classification for sufficiently differentiable coefficient functions and determine all functional forms that admit such symmetry extensions. The symmetry structure is shown to depend fundamentally on the relationship between the transport coefficient and the potential function, resulting in three distinct symmetry classes. For each case, optimal systems of one-dimensional Lie subalgebras are constructed, and the governing fractional partial differential equations are systematically reduced to fractional ordinary differential equations. Exact invariant solutions are obtained in closed form and expressed in terms of Mittag--Leffler functions, generalized Wright functions, and Fox $H$-functions. These analytical solutions provide valuable insights into fractional telegraph-type transport phenomena and serve as important benchmarks for validating numerical methods in industrial transport modeling and fractional evolution systems.
\end{abstract}

\maketitle

\section{Introduction}

The telegraph equation is a fundamental mathematical model that describes signal propagation and transport phenomena involving finite propagation speed and memory effects. Originally introduced in electrical engineering to model signal transmission in cables, telegraph equations have since found widespread applications in wave propagation, heat transport, and charge transport in complex media. Unlike classical diffusion equations, telegraph equations incorporate both wave-like and diffusive characteristics, making them particularly suitable for modeling transport processes in industrial systems and heterogeneous materials \cite{Kac1974, OrsingherBeghin2004}.

The symmetry and conservation law structures of telegraph equations have been extensively studied using the Lie group method. Bluman, Temuerchaolu, and Sahadevan established local and nonlocal symmetries for nonlinear telegraph equations and demonstrated their fundamental role in constructing exact solutions and conservation laws \cite{BlumanTemuerchaoluSahadevan2005}. Further studies have revealed deep connections between symmetry structures and conservation laws in telegraph-type equations, providing powerful analytical tools for understanding their mathematical properties and physical behavior \cite{BlumanTemuerchaolu2005JMAA, BlumanTemuerchaolu2005JMP}.

In many industrial and physical systems, transport processes exhibit memory and hereditary effects that cannot be adequately described using classical integer-order differential equations. Such phenomena arise in viscoelastic materials, anomalous heat conduction, and charge transport in disordered semiconductors, where the transport dynamics depend on the history of the system. Fractional differential equations provide a natural and effective framework for modeling these processes by incorporating fractional derivatives that capture memory and nonlocal behaviors \cite{Podlubny1999, MillerRoss1993, Mainardi2010, MetzlerKlafter2000}. Fractional transport models have been successfully applied to describe anomalous diffusion in porous media, heat transport in materials with thermal memory, and dispersive charge transport in amorphous semiconductors \cite{ScherMontroll1975, HenryWearne2000}.

Here, we consider the Riemann--Liouville fractional derivative defined by
\begin{align}\label{defrld}
\frac{\partial^\alpha}{\partial t^\alpha} u(x,t)
:=
\begin{cases}
\displaystyle \frac{\partial^n}{\partial t^n} u(x,t),
& \alpha = n \in \mathbb{N}, \\[8pt]
\displaystyle \frac{1}{\Gamma(n-\alpha)}
\frac{\partial^n}{\partial t^n}
\int_0^t \frac{u(x,s)}{(t-s)^{\alpha-n+1}} ds,
& \alpha \in (n-1,n),\; n \in \mathbb{N}.
\end{cases}
\end{align}
This definition generalizes classical integer-order differentiation to non-integer orders and naturally incorporates memory effects through the convolution integral.

Fractional generalizations of telegraph equations have attracted significant attention due to their ability to model transport phenomena in heterogeneous media and complex industrial systems. These equations naturally arise in the modeling of wave propagation in viscoelastic materials, heat conduction with memory effects, and charge transport in semiconductor devices with spatially varying material properties \cite{OrsingherBeghin2004, Mainardi2010}.

Lie symmetry analysis provides a systematic method for studying differential equations by identifying the transformation groups that leave the governing equations invariant. This approach enables symmetry classification, reduction of partial differential equations to ordinary differential equations, and construction of exact invariant solutions \cite{BlumanKumei1989, Olver1993}. The extension of Lie symmetry methods to fractional differential equations was initiated by Gazizov et al. \cite{GazizovKasatkinLukashchuk2007}, and further developments have established symmetry classifications for various classes of fractional evolution and diffusion equations \cite{LukashchukMakunin2015, HuangShen2015, ourNL, Baleanu2022, Wang2023, Gurefe2023}.

In this study, we consider the following system of time-fractional telegraph equations:
\begin{equation}\label{1}
\begin{cases}
\frac{\partial^\alpha u}{\partial t^\alpha} = v_x, \\
\frac{\partial^\alpha v}{\partial t^\alpha} = f(x)u_x + g(x)u,
\end{cases}
\end{equation}
where $\alpha > 0$ is the fractional order, and $f(x)$ and $g(x)$ are sufficiently differentiable functions with $f(x) > 0$. This system represents a fractional generalization of the classical telegraph equations with spatially varying coefficients. The function $f(x)$ characterizes spatially dependent transport properties such as diffusivity or mobility, while $g(x)$ represents external potential or forcing effects arising in inhomogeneous media.

While previous studies have investigated the symmetry properties of fractional evolution equations under restricted conditions, a complete Lie symmetry classification of fractional telegraph systems with general spatially varying coefficients remains an open problem. The primary objective of this study is to establish a complete Lie symmetry classification for system \eqref{1} and determine all coefficient functions that admit symmetry extensions.

Using Lie symmetry analysis and optimal system methods \cite{PateraWinternitz1977}, we systematically classify the symmetry structure of the system and derive exact invariant solutions expressed in terms of Mittag--Leffler functions, generalized Wright functions, and Fox $H$-functions \cite{ourFCAA}. These analytical solutions provide valuable insights into fractional telegraph-type transport processes and serve as important benchmarks for validating the numerical methods used in industrial transport modeling.

The remainder of this paper is organized as follows. In Section 2, we perform Lie symmetry analysis and derive the determining equations. Section 3 presents the construction of the optimal systems and symmetry reductions. Section 4 provides explicit invariant solutions. Finally, Section 5 summarizes the results and discusses their implications for fractional telegraph systems and industrial transport modeling.

\section{Lie symmetry analysis of the class of the
time fractional Telegraph Equations}

We first recall the necessary definitions and formulas required to perform the Lie symmetry analysis of a system of FPDEs. The general form of a system of time FPDEs with two independent variables $x$ and $t$ is as follows:
\begin{equation}\label{geneq}
\begin{cases}
\cfrac{\partial^\alpha u}{\partial t^\alpha} = F_1(x,t,u,u_x,u_{xx},\ldots, v, v_x, v_{xx},\ldots), \\[10pt]
\cfrac{\partial^\alpha v}{\partial t^\alpha} = F_2(x,t,u,u_x,u_{xx},\ldots, v, v_x, v_{xx},\ldots),
\end{cases}
\end{equation}
where the subscripts denote partial derivatives, and $\alpha$ is a positive real number. In the Lie symmetry analysis, the infinitesimal generator of (\ref{geneq}) is given by
\begin{equation*}
X=\xi\frac{\partial}{\partial x}+\tau\frac{\partial}{\partial t}+\mu\frac{\partial}{\partial u}+\phi\frac{\partial}{\partial v},
\end{equation*}
and the corresponding prolonged infinitesimal generator is
\begin{equation}\label{gener}
\tilde{X}=X+\mu^{(\alpha)}\frac{\partial}{\partial u_{t^{\alpha}}}+\mu^{(1)}\frac{\partial}{\partial u_x}+\cdots
+\phi^{(\alpha)}\frac{\partial}{\partial v_{t^{\alpha}}}+\phi^{(1)}\frac{\partial}{\partial v_x}+\cdots,
\end{equation}
where $\tau,$ $\xi,$ $\mu$ and $\phi$ are infinitesimals and $\mu^{(\alpha)},$ $\mu^{(n)},$ $\phi^{(\alpha)}$ and $\phi^{(n)}$ $(n=1,2,\ldots)$ are extended infinitesimals, $u_{t^\alpha}=\cfrac{\partial^\alpha u}{\partial t^\alpha}$, $v_{t^\alpha}=\cfrac{\partial^\alpha v}{\partial t^\alpha}$. Explicitly, $\mu^{(n)}$ and $\phi^{(n)}$ are given by
\begin{align}\label{mun}
&\mu^{(1)}  =  \mathrm{D}_x(\mu)-u_x\mathrm{D}_x(\xi)-u_t\mathrm{D}_x(\tau), \nonumber\\
&\mu^{(2)}  =  \mathrm{D}_x(\mu^{(1)})-u_{xx}\mathrm{D}_x(\xi)-u_{x t}\mathrm{D}_x(\tau),\nonumber\\
&\vdots  \nonumber\\
&\phi^{(1)} = \mathrm{D}_x(\phi)-v_x\mathrm{D}_x(\xi)-v_t\mathrm{D}_x(\tau),\nonumber\\
&\phi^{(2)} = \mathrm{D}_x(\phi^{(1)})-v_{xx}\mathrm{D}_x(\xi)-v_{x t}\mathrm{D}_x(\tau),\nonumber\\
&\vdots
\end{align}
where $\mathrm{D}_x$ is the total derivative operator defined as
\begin{equation*}
\mathrm{D}_x  :=  \frac{\partial}{\partial x}+u_x\frac{\partial}{\partial u}+u_{xx}\frac{\partial}{\partial u_x}+\cdots
+v_x\frac{\partial}{\partial v}+v_{xx}\frac{\partial}{\partial v_x}+\cdots.
\end{equation*}
The $\alpha$th order extended infinitesimals have the following forms \cite{ourNL}:
\begin{eqnarray}\label{mual2}
\mu^{(\alpha)} & = & \frac{\partial^\alpha\mu}{\partial t^\alpha}-u\frac{\partial^\alpha\mu_u}{\partial t^\alpha}-v\frac{\partial^\alpha\mu_v}{\partial t^\alpha}+\left(\mu_u-\alpha \mathrm{D}_t(\tau)\right)\frac{\partial^\alpha u}{\partial t^\alpha}+\mu_v\frac{\partial^\alpha v}{\partial t^\alpha}-\sum_{n=1}^\infty\binom{\alpha}{n}\mathrm{D}_t^n(\xi)\mathrm{D}_t^{\alpha-n}(u_x)\nonumber\\
& & +\sum_{n=1}^\infty\left[\binom{\alpha}{n}\frac{\partial^n\mu_u}{\partial t^n}-\binom{\alpha}{n+1}\mathrm{D}_t^{n+1}(\tau)\right]\mathrm{D}_t^{\alpha-n}(u)+\sum_{n=1}^\infty\binom{\alpha}{n}\frac{\partial^n\mu_v}{\partial t^n}\mathrm{D}_t^{\alpha-n}(v)+\mu_1,
\end{eqnarray}
\begin{eqnarray}\label{phial}
\phi^{(\alpha)} & = & \frac{\partial^\alpha\phi}{\partial t^\alpha}-v\frac{\partial^\alpha\phi_v}{\partial t^\alpha}-u\frac{\partial^\alpha\phi_u}{\partial t^\alpha}+(\phi_v-\alpha \mathrm{D}_t(\tau))\frac{\partial^\alpha v}{\partial t^\alpha}+\phi_u\frac{\partial^\alpha u}{\partial t^\alpha}-\sum_{n=1}^\infty\binom{\alpha}{n}\mathrm{D}_t^n(\xi)\mathrm{D}_t^{\alpha-n}(v_x)\nonumber\\
& & +\sum_{n=1}^\infty\left[\binom{\alpha}{n}\frac{\partial^n\phi_v}{\partial t^n}-\binom{\alpha}{n+1}\mathrm{D}_t^{n+1}(\tau)\right]\mathrm{D}_t^{\alpha-n}(v)+\sum_{n=1}^\infty\binom{\alpha}{n}\frac{\partial^n\phi_u}{\partial t^n}\mathrm{D}_t^{\alpha-n}(u)+\phi_1,
\end{eqnarray}
where
\begin{eqnarray*}
\mu_1 & = & \sum_{n=2}^\infty\sum_{m_1+m_2=2}^n\sum_{\begin{subarray}
~k_1=0,..,m_1\\
k_2=0,..,m_2\\
k_1+k_2\geq 2
\end{subarray}}\sum_{r_1=0}^{k_1}\sum_{r_2=0}^{k_2}\binom{\alpha}{n}\binom{n}{m_1}\binom{n-m_1}{m_2}\binom{k_1}{r_1}\binom{k_2}{r_2} \frac{1}{k_1!k_2!}\frac{t^{n-\alpha}}{\Gamma(n+1-\alpha)}\\
&&\times(-u)^{r_1}(-v)^{r_2}\frac{\partial^{m_1}u^{k_1-r_1}}{\partial t^{m_1}}\frac{\partial^{m_2}v^{k_2-r_2}}{\partial t^{m_2}} \frac{\partial^{n-m_1-m_2+k_1+k_2}\mu}{\partial t^{n-m_1-m_2}\partial u^{k_1}\partial v^{k_2}},\nonumber
\end{eqnarray*}
and
\begin{eqnarray*}
\phi_1 &=& \sum_{n=2}^\infty\sum_{m_1+m_2=2}^n\sum_{\begin{subarray}
~k_1=0,..,m_1\\
k_2=0,..,m_2\\
k_1+k_2\geq 2
\end{subarray}}\sum_{r_1=0}^{k_1}\sum_{r_2=0}^{k_2}\binom{\alpha}{n}\binom{n}{m_1}
 \binom{n-m_1}{m_2}\binom{k_1}{r_1}\binom{k_2}{r_2} \frac{1}{k_1!k_2!}
\frac{t^{n-\alpha}}{\Gamma(n+1-\alpha)}\\
&&\times (-u)^{r_1}(-v)^{r_2}\frac{\partial^{m_1}u^{k_1-r_1}}{\partial t^{m_1}}\frac{\partial^{m_2}v^{k_2-r_2}}{\partial t^{m_2}}
\frac{\partial^{m-m_1-m_2+k_1+k_2}\phi}{\partial t^{n-m_1-m_2}\partial u^{k_1}\partial v^{k_2}}.
\end{eqnarray*}
Here, $\mathrm{D}_t$ is the total derivative operator defined as
\begin{equation*}
\mathrm{D}_t  :=  \frac{\partial}{\partial t}+u_t\frac{\partial}{\partial u}+u_{xt}\frac{\partial}{\partial u_x}+\cdots
+v_t\frac{\partial}{\partial v}+v_{xt}\frac{\partial}{\partial v_x}+\cdots.
\end{equation*}
Because the lower limit of the integral in (\ref{defrld}) is fixed, it should be invariant with respect to point transformations. We thus arrive at the initial condition
\begin{equation}\label{incon}
\tau(x,t,u,v)|_{t=0}=0.
\end{equation}
Note: If $\mu$ or $\phi$ is linear in $u$ and $v,$ then $\mu_1=0$ and $\phi_1=0,$ respectively.

The infinitesimal invariance criterion in the Lie symmetry analysis for the system given in (\ref{geneq}) is
\begin{equation}\label{gende}
\begin{cases}
\left.\tilde{X}\left(\frac{\partial^\alpha u}{\partial t^\alpha}-F_1(x,t,u,u_x,u_{xx},\ldots, v, v_x, v_{xx},\ldots)\right)\right|_{(\ref{geneq})}=0,\\
\left.\tilde{X}\left(\frac{\partial^\alpha v}{\partial t^\alpha}-F_2(x,t,u,u_x,u_{xx},\ldots, v, v_x, v_{xx},\ldots)\right)\right|_{(\ref{geneq})}=0,
\end{cases}
\end{equation}
where $\tilde{X}$ is given by (\ref{gener}), (\ref{mun}), (\ref{mual2}) and (\ref{phial}). We are now ready to investigate the infinitesimal symmetries of the system in (\ref{1}) following the aforementioned Lie symmetry analysis of the time fractional system.
\begin{theorem}\label{thm_symm}
Let $g(x)\neq 0$ and $\omega_{\lambda_1}(x)=\int_\beta^x \frac{dr}{\sqrt{f(r)}} +\lambda_1$, where $\lambda_1\in \mathbb{R}$ and $\beta$ is a suitable real number. Then all of the infinitesimal symmetries of Eq.~(\ref{1}) are given by the following cases:
\begin{enumerate}[label=(\roman*)]
    \item  for all forms of $f(x)$ and $ g(x)$ except the following three cases, it has the following infinitesimal symmetries
    \begin{equation*}
    X_0 = c_1(x,t)\frac{\partial}{\partial u}+c_2(x,t)\frac{\partial}{\partial v}~\mbox{ and }
    X_1 = u\frac{\partial}{\partial u}+v\frac{\partial}{\partial v},
    \end{equation*}
    where $c_1(x,t)$ and $c_2(x,t)$ are solutions of (\ref{1}).
    \item  for any $f(x)$ and $g(x)=\frac{\lambda_2 \sqrt{f(x)}}{\omega_{\lambda_1}(x)}+\frac{f(x)'}{2}$ (here $\lambda_2\in \mathbb{R},~ \lambda_2\ne 0$), it has the infinitesimal symmetries $X_0$, $X_1$ and
    \begin{equation*}
   X_2 = \sqrt{f(x)}\omega_{\lambda_1}(x)\frac{\partial}{\partial x}
+\frac{t}{\alpha}\frac{\partial}{\partial t}-\frac{f(x)'\omega_{\lambda_1}(x)}{2\sqrt{f(x)}}u \frac{\partial}{\partial u}.
    \end{equation*}
    \item for any $f(x)$ and $g(x)=\lambda_2 \sqrt{f(x)}+\frac{f(x)'}{2}, ~\lambda_2$ (here $\lambda_2\in \mathbb{R},~ \lambda_2\ne 0$), it has the infinitesimal symmetries $X_0$, $X_1$ as well as
    \begin{equation*}
    X_3=\sqrt{f(x)}\frac{\partial}{\partial x}-\frac{f(x)'}{2\sqrt{f(x)}}u\frac{\partial}{\partial u}.
    \end{equation*}
    \item for any $f(x)$ and $g(x)=\frac{f(x)'}{2}$ it has the symmetries $X_0$, $X_1$ along with
    \begin{eqnarray*}
     X_4&=&\sqrt{f(x)}\omega_{\lambda_1}(x) \frac{\partial}{\partial x}+\frac{t}{\alpha}\frac{\partial}{\partial t}-\frac{f(x)'}{2\sqrt{f(x)}}\omega_{\lambda_1}(x) u 
     \frac{\partial}{\partial u},\\
     X_5&=&\sqrt{f(x)}\frac{\partial}{\partial x}-\frac{f(x)'}{2\sqrt{f(x)}}u\frac{\partial}{\partial u},\\
    X_6&=&\frac{1}{\sqrt{f(x)}}v\frac{\partial}{\partial u}+\sqrt{f(x)}u\frac{\partial}{\partial v}.
    \end{eqnarray*}
\end{enumerate}
\end{theorem}


\begin{proof}
To find the Lie symmetries, first we need to find a solution to the following determining system
\begin{equation*}
\begin{cases}
\tilde{X}\left(\frac{\partial^\alpha u}{\partial t^\alpha}-v_x\right)|_{(\ref{1})}=0,\\
\tilde{X}(\frac{\partial^\alpha v}{\partial t^\alpha}-f(x)u_x-g(x)u)|_{(\ref{1})}=0.
\end{cases}
\end{equation*}
In explicit form, this is
\begin{equation}\label{10}
\begin{cases}
\left.\left(\mu^{(\alpha)}-\phi^{(1)}\right)\right|_{(\ref{1})}=0,\\
\left.\left(\phi^{(\alpha)}- f(x)' \xi u_x-f(x)\mu^{(1)}-g(x)'u\xi-g(x)\mu\right)\right|_{(\ref{1})}=0.
\end{cases}
\end{equation}
From (\ref{10}), we obtain the following (overdetermined) system of determining equations by setting the coefficients of the linearly independent partial derivatives $D_t^{\alpha-n}u,$ $D_t^{\alpha-n}v,$ $D_t^{\alpha-n}u_x,$ $D_t^{\alpha-n}v_x,$ $v_x,$ $u_x,$ $v_t,$ $u_x v_t,$ $v_x v_t,$ $u_x v_x$, $v_x^2$ and $1$ equal to zero:
\begin{eqnarray}
& & \tau_x=\tau_u=\tau_v=\xi_u=\xi_v=0,\label{aneg}\\
& & D_t^n(\xi)=0,\qquad n=1,2,\ldots,\label{gurav}\\
& & \mu_u-\alpha D_t(\tau)-\phi_v+\xi_x=0,\label{duruv}\\
& & f(x)\phi_v-\alpha f(x) D_t(\tau)-f(x)'\xi-f(x)\mu_u+f(x)\xi_x=0,\label{es}\\
& & \binom{\alpha}{n}\frac{\partial^n\mu_u}{\partial t^n}-\binom{\alpha}{n+1}D_t^{n+1}(\tau)=0,\qquad n=1,2,\ldots,\label{neg}\\
& & \frac{\partial^n\mu_v}{\partial t^n}=0,\qquad n=1,2,\ldots,\label{hoyor}\\
& & f(x)\mu_v-\phi_u=0,\label{tav}\\
& & \frac{\partial^n\phi_u}{\partial t^n}=0,\qquad n=1,2,\ldots,\label{doloo}\\
& & \binom{\alpha}{n}\frac{\partial^n\phi_v}{\partial t^n}-\binom{\alpha}{n+1}D_t^{n+1}(\tau)=0,\qquad n=1,2,\ldots,\label{naim}\\
& & \frac{\partial^\alpha\mu}{\partial t^\alpha}-u\frac{\partial^\alpha\mu_u}{\partial t^\alpha}-v\frac{\partial^\alpha\mu_v}{\partial t^\alpha}+g(x)u\mu_v-\phi_x=0,\label{zurgaa}
\end{eqnarray}
and
\begin{equation}\label{arav}
\frac{\partial^\alpha\phi}{\partial t^\alpha}-v\frac{\partial^\alpha\phi_v}{\partial t^\alpha}-u\frac{\partial^\alpha\phi_u}{\partial t^\alpha}+g(x)u\phi_v-\alpha g(x)uD_t(\tau)-g(x)'u\xi-g(x)\mu-f(x)\mu_x=0.
\end{equation}

Analyzing (\ref{aneg}), (\ref{gurav}), (\ref{duruv}) and (\ref{es}) with the initial condition $\tau(x,t,u,v)\vert_{t=0}=0$, we deduce the following infinitesimal symmetries:
\begin{eqnarray}
\tau(t)&=&\frac{s_1}{\alpha}t,\label{eq:tau1}\\
\xi(x)&=&\sqrt{f(x)}\left( s_1\int^x_{\beta} \frac{1}{\sqrt{f(r)}}dr+s_2 \right)\label{xi}
\end{eqnarray} 
for some $s_1$ and $s_2$.

The following lemma establishes that the infinitesimals $\mu(x, t, u, v)$ and $\phi(x, t, u, v)$ are necessarily linear with respect to the dependent variables $u$ and $v$. This result is crucial as it ensures that the higher-order supplementary terms $\mu_1$ and $\phi_1$ in the infinitesimal prolongations \eqref{mual2} and \eqref{phial} vanish identically, thereby significantly simplifying the system of determining equations.

\begin{lemma} \label{lemma:linear_inf}
 $\mu(x,t,u,v)$ and $\phi(x,t,u,v)$ are linear functions with respect to $u$ and $v$. Thus, $\mu_1=0$ and $\phi_1=0$.
\end{lemma}

\begin{proof}
From \eqref{duruv} and \eqref{tav} by differentiating them with respect to $u$ and $v$ we obtain
\begin{align}
\mu_{uu} - \phi_{vu} &= 0, \label{11u} \\
\mu_{uv} - \phi_{vv} &= 0, \label{11v} \\
f(x)\mu_{vu} - \phi_{uu} &= 0,\label{15u}\\
f(x)\mu_{vv} - \phi_{uv} &= 0 \label{15v}.
\end{align}
By subtracting \eqref{15v} from \eqref{11u}, and by subtracting \eqref{15u} from the product of $f(x)$ and \eqref{11v},  we arrive at the following relations
\begin{align}
\mu_{uu} -  f(x) \mu_{vv} &= 0, \label{11u15v} \\
\phi_{uu} - f(x)\phi_{vv} &= 0 \label{f11v15u}. 
\end{align}
By taking second-order derivative from \eqref{zurgaa} with respect $u$ and $v$ we obtain
\begin{align}
-\frac{\partial^\alpha\mu_{uu}}{\partial t^\alpha}-u\frac{\partial^\alpha\mu_{uuu}}{\partial t^\alpha}-v\frac{\partial^\alpha\mu_{uuv}}{\partial t^\alpha}+2g(x)\mu_{uv}+g(x)u\mu_{uuv}-\phi_{xuu}=0 \label{19uu}, \\
-u\frac{\partial^\alpha\mu_{uvv}}{\partial t^\alpha}-\frac{\partial^\alpha\mu_{vv}}{\partial t^\alpha}-v\frac{\partial^\alpha\mu_{vvv}}{\partial t^\alpha}+g(x)u\mu_{vvv}-\phi_{xvv}=0
\label{19vv} 
\end{align}
correspondingly. 

Next, by subtracting the product of $f(x)$ and \eqref{19vv} from \eqref{19uu}, and utilizing the identities \eqref{11u15v} and \eqref{f11v15u}, we obtain
\begin{equation}
2g(x)\mu_{uv} = 0.
\end{equation}
Given that $g(x) \neq 0$, it follows that $\mu_{uv} = 0$. Applying this result, equation \eqref{11v} reduces to $\phi_{vv} = 0$ and \eqref{15u} reduces to $\phi_{uu} = 0$. Consequently, $\phi(x, t, u, v)$ is established as a linear function with respect to the dependent variables $u$ and $v$.

In an analogous manner, we consider the second-order derivatives of \eqref{arav} with respect to the dependent variables. By subtracting the product of $f(x)$ and the second-order derivative of \eqref{arav} with respect to $v$ from its second-order derivative with respect to $u$, we obtain
\begin{equation}
2g(x)\phi_{uv} = 0.
\end{equation}
Since $g(x) \neq 0$, it follows that $\phi_{uv} = 0$. Utilizing the previously established relations in \eqref{11v} and \eqref{15u}, we find that $\mu_{vv} = 0$. This confirms that $\mu(x, t, u, v)$ is also a linear function with respect to $u$ and $v$.

\end{proof}

Using Lemma \ref{lemma:linear_inf}, (\ref{duruv}), (\ref{es}), (\ref{neg}), (\ref{hoyor}), (\ref{tav}), (\ref{doloo}) and (\ref{naim}), we have the following form of $\mu$ and $\phi$
\begin{align}
\mu(x,t,u,v)&=\left[a_1(x)-\frac{f(x)'}{2\sqrt{f(x)}}\left( s_1\int^x_{\beta} \frac{1}{\sqrt{f(r)}}dr+s_2 \right)\right]u+a_2(x)v+c_1(x,t),\label{eq:mu1}\\
\phi(x,t,u,v)&=a_2(x)f(x)u+a_1(x)v+c_2(x,t)\label{eq:phi1}
\end{align}
for some functions $a_1(x),$ $a_2(x)$, $c_1(x,t)$ and $c_2(x,t)$.

Substituting (\ref{eq:tau1}), (\ref{eq:mu1}) and (\ref{eq:phi1}) into (\ref{zurgaa}) and  (\ref{arav}), we obtain the following system by setting the coefficients of $1$, $u$ and $v$ of (\ref{zurgaa}) and  (\ref{arav})  equal to zero:
\begin{align}
&\partial^\alpha_t c_1(x,t)-c_2(x,t)_x=0,\label{u1}\\
&\partial^\alpha_tc_2(x,t)-f(x)c_1(x,t)_x-g(x)c_1(x,t)=0,\label{u2}\\
&a_1(x)_x=0,\label{u3}\\
&f(x)a_2(x)_x+g(x)a_2(x)=0,\label{u4}\\
&g(x)a_2(x)-\left(f(x)a_2(x)\right)_x=0,\label{u5}\\
&\left(k(x)'\int^x_{\beta} \frac{1}{\sqrt{f(r)}}dr+\frac{k(x)}{\sqrt{f(x)}}\right)s_1+k(x)'s_2-a_1(x)_x=0,\label{u6}
\end{align}
where $k(x)=\left(\sqrt{f(x)}\right)'-\frac{g(x)}{\sqrt{f(x)}}.$ From (\ref{u1}) and (\ref{u2}), the functions $c_1(x,t)$ and $c_2(x,t)$ are a solution of (\ref{1}), and from (\ref{u3}) we get  $a_1(x)=s_3,$ $s_3\in \mathbb{R}$. Therefore, the system in (\ref{1}) admits the following infinitesimal symmetries for any $f(x)$ and $g(x)$:
$$
X_0=c_1(x,t)\frac{\partial}{\partial u}+c_2(x,t)\frac{\partial}{\partial v}~\text{ and }~X_1=u\frac{\partial}{\partial u}+v\frac{\partial}{\partial v}.
$$
Obviously, the infinitesimal symmetry $X_0$ implies nothing else but the linearity of the original system. 

By adding (\ref{u4}) and (\ref{u5}), we get $(f(x)'-2g(x))a_2(x)=0$. Hence, we consider two cases characterized by the relations $f(x)'-2g(x)\neq 0$ and $f(x)'-2g(x)=0$.

\subsection*{Case 1}
Let us consider the case $f(x)'-2g(x)\neq 0$. Then $a_2(x)=0$, $k(x)\neq0$ and (\ref{u6}) become 
\begin{equation}\label{eq:aravus1s2}
\left(k(x)'\int^x_{\beta} \frac{1}{\sqrt{f(r)}}dr+\frac{k(x)}{\sqrt{f(x)}}\right)s_1
=-k(x)'s_2,
\end{equation}
where $k(x)=\left(\sqrt{f(x)}\right)'-\frac{g(x)}{\sqrt{f(x)}}.$

In view of the above relation, it is possible that the system in Eq.(\ref{1}) admits a symmetry other than $X_0$ and $X_1$ if and only if $s_1\neq 0$ or $s_2\neq 0$. Hence, the following two subcases arise depending on $k(x)'\neq 0$ and $k(x)'=0$.

\subsubsection*{Case 1.1}
Let us assume $k(x)'\neq 0.$ Then, there exists $\lambda_1\in\mathbb{R}$ such that $s_2=s_1\lambda_1$ and we see from (\ref{eq:aravus1s2}) that the functions $f(x)$ and $g(x)$ should satisfy the following relation 
$$g(x)=\frac{\lambda_2 \sqrt{f(x)}}{\omega_{\lambda_1}(x)}+\frac{f(x)'}{2},~\omega_{\lambda_1}(x)=\int_\beta^x \frac{dr}{\sqrt{f(r)}} +\lambda_1,~\lambda_2\in \mathbb{R},~ \lambda_2\ne 0.$$ 
Hence, for any $f(x)$ and $g(x)=\frac{\lambda_2 \sqrt{f(x)}}{\omega_{\lambda_1}(x)}+\frac{f(x)'}{2}$ ( here $\omega_{\lambda_1}(x)=\int_\beta^x \frac{dr}{\sqrt{f(r)}} +\lambda_1$, $\lambda_2\in \mathbb{R}$, $\lambda_2\ne 0$),  the system (\ref{1}) has the following additional symmetry
\begin{equation*}
X_2 = \sqrt{f(x)}\omega_{\lambda_1}(x)\frac{\partial}{\partial x}
+\frac{t}{\alpha}\frac{\partial}{\partial t}-\frac{f(x)'}{2\sqrt{f(x)}}\omega_{\lambda_1}(x)u \frac{\partial}{\partial u}.
\end{equation*}

\subsubsection*{Case 1.2}
Suppose that $k(x)'=0$, it immediately follows that
$$g(x)=\lambda_2 \sqrt{f(x)}+\frac{f(x)'}{2}, ~\lambda_2\in \mathbb{R},~ \lambda_2\ne 0.$$
So, for any $f(x)$ and $g(x)=\lambda_2 \sqrt{f(x)}+\frac{f(x)'}{2}$ (here $\lambda_2\in \mathbb{R},~ \lambda_2\ne 0$), the system (\ref{1}) has the following additional symmetry
\begin{equation*}
X_3=\sqrt{f(x)}\frac{\partial}{\partial x}-\frac{f(x)'}{2\sqrt{f(x)}}u\frac{\partial}{\partial u}.
\end{equation*}

\subsection*{Case 2}
Next, let us consider the case $f(x)'-2g(x)=0$. Since $k(x)=0$, in this case, we can see that $f$ and $g$ satisfy (\ref{u6}). Also,  by (\ref{u4}) and (\ref{u5}), we get $a_2(x)=\frac{s_4}{\sqrt{f(x)}},$ $s_4\in \mathbb{R}$. Hence, for any $f(x)$ and $g(x)=\frac{f(x)'}{2}$, the system (\ref{1}) has the following three additional symmetries
\begin{eqnarray*}
X_4&=&\sqrt{f(x)}\omega_{\lambda_1}(x)\frac{\partial}{\partial x}+\frac{t}{\alpha}\frac{\partial}{\partial t}-\frac{f(x)'}{2\sqrt{f(x)}}\omega_{\lambda_1}(x) u 
\frac{\partial}{\partial u},\\
X_5&=&\sqrt{f(x)}\frac{\partial}{\partial x}-\frac{f(x)'}{2\sqrt{f(x)}}u\frac{\partial}{\partial u},\\    X_6&=&\frac{1}{\sqrt{f(x)}}v\frac{\partial}{\partial u}+\sqrt{f(x)}u\frac{\partial}{\partial v}.
\end{eqnarray*}
\end{proof}

We are ready to explore the one-dimensional optimal systems of Lie algebras of infinitesimal symmetries for individual cases and the classification of group invariant solutions, given that the system in Eq.(\ref{1}) has a complete group classification.

In the subsequent calculations, we ignored the trivial infinitesimal symmetry $X_0.$ In the Case i., the Lie algebra is generated by only $X_1,$ which does not yield any invariant solutions. Therefore, we will consider only Cases ii. through iv..

\section{Optimal systems of subalgebras of Lie algebras of the infinitesimal symmetries}\label{sec:RS}

The problem of enumeration of all subalgebras of a given finite-dimensional Lie algebra is vital for the group analysis of differential equations. To do so, a group of automorphisms is calculated using the algorithm presented in \cite{PateraWinternitz1977}. Every subalgebra transforms into a subalgebra of equal dimension under an automorphism. Hence, all subalgebras of the given Lie algebra are decomposed into classes of similar subalgebras. The union of the representatives from every class of subalgebras is called the optimal system of the Lie algebra. Therefore, the classification problem is reduced to the problem of constructing optimal systems of the Lie algebra of the Lie group. The optimal systems of Lie algebras with dimensions less than 6 were calculated in \cite{PateraWinternitz1977}.
 
In this section, we classify all the distinctly different invariant solutions by constructing optimal systems of one-dimensional Lie subalgebras of the Lie algebra. The search for this type of solutions reduces the number of variables. Thus, we derived reduced fractional ordinary differential equations. 

\textbf{Case 1. Reduced equations of (\ref{1}) with any $f(x)$ and $g(x)=\frac{\lambda_2\sqrt{f(x)}}{\omega_{\lambda_1}(x)}+\frac{f(x)'}{2}$ (here $\omega_{\lambda_1}(x)=\int_\beta^x \frac{dr}{\sqrt{f(r)}} +\lambda_1$, $\lambda_1,\lambda_2\in \mathbb{R},~ \lambda_2\ne 0$)} 

Let us consider the following telegraph system
\begin{equation}
\label{eq:TEcaseA2}
\begin{cases}
\frac{\partial^\alpha u}{\partial t^\alpha}=v_x,\\
\frac{\partial^\alpha v}{\partial t^\alpha}=f(x) u_x+
\left (\frac{\lambda_2\sqrt{f(x)}}{\omega_{\lambda_1}(x)}+\frac{f(x)'}{2}\right)u,
\end{cases}
\end{equation}
where $f(x)$ is a sufficiently differentiable function, $\omega_{\lambda_1}(x)=\int_\beta^x \frac{dr}{\sqrt{f(r)}} +\lambda_1$ and $\lambda_1,\lambda_2\in \mathbb{R},~ \lambda_2\ne 0$. The corresponding infinitesimal symmetries are obtained
\begin{equation*}
 X_1=u\frac{\partial}{\partial u}+v\frac{\partial}{\partial v},
 \end{equation*}
 \begin{equation*}
 X_2=\sqrt{f(x)}\omega_{\lambda_1}(x)\frac{\partial}{\partial x}
    +\frac{t}{\alpha}\frac{\partial}{\partial t}-\frac{f(x)'}{2\sqrt{f(x)}}\omega_{\lambda_1}(x)u \frac{\partial}{\partial u}.
\end{equation*}
The Lie algebra generated by $X_1$ and $X_2$ is Abelian, and thus the optimal system consists of
$$W_1=X_1,\quad W_2=X_2+aX_1,\quad\text{where }a\in \mathbb{R}.$$
For $W_2$, the characteristic equation reads as
\begin{eqnarray}
\frac{dx}{\sqrt{f(x)}\omega_{\lambda_1}(x)}=
\frac{\alpha dt}{t}=
\frac{du}{\left(a-\frac{f(x)'}{2\sqrt{f(x)}}\omega_{\lambda_1}(x)\right)u}
=\frac{dv}{av},
\end{eqnarray}
which gives the similarity variable
$
z=\omega_{\lambda_1}(x)^{-\frac{1}{\alpha}}t.
$ Thus, the similarity transformation is
\begin{eqnarray}
\label{eq:ST_thm_ii}
\begin{cases}
u(x,t)=\frac{1}{\sqrt{f(x)}}\omega_{\lambda_1}(x)^{a}\varphi(z),\\
v(x,t)=\omega_{\lambda_1}(x)^a\psi(z).
\end{cases}
\end{eqnarray}
Substituting (\ref{eq:ST_thm_ii}) into (\ref{eq:TEcaseA2}), we get the following reduced system
\begin{eqnarray}
\label{eq:RS_thm_ii}
\begin{cases}
\frac{d^\alpha}{dz^\alpha}\varphi(z)=a\psi(z)-\frac{1}{\alpha}z\psi(z)',\\
\frac{d^\alpha}{dz^\alpha}\psi(z)=(a+\lambda_2)\varphi(z)-\frac{1}{\alpha}z\varphi(z)'.
\end{cases}
\end{eqnarray}

The infinitesimal symmetry $W_1$ appears in each optimal system in the following subsections, and it does not yield any invariant solutions.

\textbf{Case 2. Reduced equations of (\ref{1}) with any $f(x)$ and $g(x)=\lambda_2 \sqrt{f(x)}+\frac{f(x)'}{2}$ (here $\lambda_2\in \mathbb{R},~ \lambda_2\ne 0$)} 

Let us consider the following telegraph equation
\begin{eqnarray}
\label{eq:TEcaseA3}
\begin{cases}
\frac{\partial^\alpha u}{\partial t^\alpha}=v_x,\\
\frac{\partial^\alpha v}{\partial t^\alpha}=f(x) u_x+
\left (\lambda_2\sqrt{f(x)}+\frac{f(x)'}{2}\right)u,
\end{cases}
\end{eqnarray}
where $f(x)$ is a  sufficiently differentiable function, and $\lambda_2\in \mathbb{R},~ \lambda_2\ne 0$. The corresponding infinitesimal symmetries are 
$$X_1=u\frac{\partial}{\partial u}+v\frac{\partial}{\partial v},\quad   X_3=\sqrt{f(x)}\frac{\partial}{\partial x}-\frac{f(x)'}{2\sqrt{f(x)}}u\frac{\partial}{\partial u}.$$
The commutator of the infinitesimal symmetries is zero, i.e., $[X_1,X_3] = 0$, as in the previous case, the optimal system is
$$W_1=X_1,\quad
W_3=X_3+aX_1,\quad\text{where } a\in \mathbb{R}.
$$
Using the characteristic method with $W_3$, we have the following invariant solutions
\begin{eqnarray}
\label{eq:ST_thm_iii}
\begin{cases}
u(x,t)=\frac{1}{\sqrt{f(x)}}\exp\left(a\omega_{0}(x)\right)\varphi(t),\\
v(x,t)=\exp\left(a\omega_{0}(x)\right)\psi(t),
\end{cases}
\end{eqnarray}
then the reduced system is obtained as
\begin{equation}
\label{eq:RS_thm_iii}
    \begin{cases}
    \frac{d^\alpha }{dt^\alpha}\varphi(t)=a\psi(t),\\
    \frac{d^\alpha }{dt^\alpha}\psi(t)=(a+\lambda_2)\varphi(t),
    \end{cases}
\end{equation}
where $\omega_{0}(x)=\int^x_{\beta} \frac{1}{\sqrt{f(r)}}dr$.

\textbf{Case 3. Reduced equations of (\ref{1}) with any $f(x)$ and $g(x)=\frac{f(x)'}{2}$} 
Let us consider the following telegraph equation
\begin{equation}\label{usys1}
    \begin{cases}
    \frac{\partial^\alpha u}{\partial t^\alpha}=v_x,\\
\frac{\partial^\alpha v}{\partial t^\alpha}=f(x) u_x+
\frac{f(x)'}{2}u,
    \end{cases}
\end{equation}
where $f(x)$ is a sufficiently differentiable function. To simplify the calculations, we chose the basis 
\begin{eqnarray}
    \nonumber
    V_1=&-X_4&=-\sqrt{f(x)}\omega_{0}(x) \frac{\partial}{\partial x}-\frac{t}{\alpha}\partial_t+\frac{f(x)'}{2\sqrt{f(x)}}\omega_{0}(x)u 
    \frac{\partial}{\partial u},\\
    \nonumber
    V_2=&-X_5&=-\sqrt{f(x)}\frac{\partial}{\partial x}+\frac{f(x)'}{2\sqrt{f(x)}}u\frac{\partial}{\partial u},\\
    \nonumber
     V_3=&X_1&=u\frac{\partial}{\partial u}+v\frac{\partial}{\partial v},\\
    \nonumber
    V_4=&X_6&=\frac{1}{\sqrt{f(x)}}v\frac{\partial}{\partial u}+\sqrt{f(x)}u\frac{\partial}{\partial v},
\end{eqnarray}
where $\omega_{0}(x)=\int^x_{\beta} \frac{1}{\sqrt{f(r)}}dr$.
The commutator table for the Lie algebra is given in Table \ref{Table1} (where indices $i$ and $j$ are for the row and column, respectively).
\begin{table}[!ht]
\centering
\begin{tabular}{|c|c|c|c|c|}
 \hline
 $[V_i, V_j]$ & $V_1$ & $V_2$ & $V_3$ & $V_4$\\
 \hline
 $V_1$ & 0 & $V_2$ & 0 & 0 \\
 \hline
 $V_2$ &$-V_2$ & 0 & 0 & 0 \\
 \hline
 $V_3$ & 0 & 0 & 0 & 0 \\
 \hline
 $V_4$ & 0 & 0 & 0 & 0 \\ \hline
\end{tabular}
\caption{Commutator table for $V_1,$ $V_2,$ $V_3$ and $V_4$.}
\label{Table1}
\end{table}

From Table \ref{Table1} that the Lie algebra is identical to the Lie algebra $A_2\oplus 2A_1$ given in \cite{PateraWinternitz1977}. 
Thus, the one-dimensional optimal system of the Lie algebra generated by $V_1 , V_2, V_3$ and $V_4$ is that obtained in 
\cite{PateraWinternitz1977}
\begin{eqnarray*}
\label{ivU1}
W_1 & = & V_3=X_1,\\
W_4 & = & V_1- a_1 V_3-a_2 V_4=-a_1X_1-X_4-a_2X_6, \text{ where } a_1,a_2\in\mathbb{R},\\
\label{ivU2}
W_5 & = & V_2-a_1 V_3-a_2 V_4=-a_1X_1-X_5-a_2X_6,\\
&&\text{ where }  (a_1, a_2)\in\{(\pm 1,a), (0,\pm 1), (0,0)|a\in\mathbb{R}\},\\
\label{ivU3}
W_6 & = & a_1 V_3+ V_4=a_1X_1+X_6, \text{ where } a_1\in\mathbb{R}.\\
\label{ivU4}
\end{eqnarray*}
For $W_4$, the \textit{invariance surface condition} reads as
\begin{eqnarray}
\begin{cases}
\sqrt{f(x)}\omega_{0}(x)u_x+\frac{t}{\alpha}u_t=
\left(a_1-\frac{f(x)'}{2\sqrt{f(x)}}\omega_{0}(x) \right) u+\frac{a_2}{\sqrt{f(x)}}v\\
\sqrt{f(x)}\omega_{0}(x) v_x+\frac{t}{\alpha}v_t=
a_1v+a_2\sqrt{f(x)}u,
\end{cases}
\end{eqnarray}
which gives the similarity variable
$z=t\omega_{0}(x)^{-\frac{1}{\alpha}}
$. Consequently, the similarity transformation is obtained as
\begin{eqnarray}
\label{eq:ST_thm_iv_U1}
\begin{cases}
u(x,t)=\frac{1}{\sqrt{f(x)}}\omega_{0}(x)^{a_1+a_2}\varphi(z)+\frac{1}{\sqrt{f(x)}}\omega_{0}(x)^{a_1-a_2}\psi(z)\\
v(x,t)=\omega_{0}(x)^{a_1+a_2}\varphi(z)-\omega_{0}(x)^{a_1-a_2}\psi(z)
\end{cases}
\end{eqnarray}
Substituting (\ref{eq:ST_thm_iv_U1}) into (\ref{usys1}), we get the reduced system
\begin{eqnarray}
\label{eq:BU1ODE}
\begin{cases}
\frac{d^{\alpha}}{dz^{\alpha}}\varphi(z)=(a_2+a_1)\varphi(z) -\frac{1}{\alpha}z\varphi(z)'\\
\frac{d^{\alpha}}{dz^{\alpha}}\psi(z)=(a_2-a_1)\psi(z) +\frac{1}{\alpha}z\psi(z)'
\end{cases}
\end{eqnarray}
Analogously to the case $W_4$, we have the following invariant solutions for $W_5$
\begin{align}\label{eq:ST_thm_iv_U2}
\begin{cases}
u(x,t)=&\frac{1}{\sqrt{f(x)}}\exp \left((a_1+a_2)\omega_{0}(x)\right)\varphi(t)\\
&+\frac{1}{\sqrt{f(x)}}\exp \left((a_1-a_2)\omega_{0}(x)\right)\psi(t)
\\
v(x,t)=&\exp \left((a_1+a_2)\omega_{0}(x)\right)\varphi(t)-\exp \left((a_1-a_2)\omega_{0}(x)\right)\psi(t).
\end{cases}
\end{align}
and the reduced system
\begin{eqnarray}
\label{eq:BU2ODE}
\begin{cases}
\frac{d^{\alpha}}{dt^{\alpha}}\varphi(t)=(a_2+a_1)\varphi(t), \\
\frac{d^{\alpha}}{dt^{\alpha}}\psi(t)=(a_2-a_1)\psi(t).
\end{cases}
\end{eqnarray}
There are no invariant solutions corresponding to $W_6$.

\section{Group invariant solutions of the system given in (\ref{1})}

In this section, we provide solutions to the reduced systems obtained in the previous section, and using these solutions, we express invariant solutions of (\ref{1}) explicitly.
The reduced systems (\ref{eq:RS_thm_ii}) and (\ref{eq:RS_thm_iii})  have the following general form:
\begin{equation}\label{e15}
\begin{cases}
\displaystyle{\frac{d^\alpha}{d z^\alpha}}\varphi(z) = a_1\psi(z)+\frac{b_1}{\alpha}z \psi(z)',\\
\displaystyle{\frac{d^\alpha}{d z^\alpha}}\psi(z) = a_2\varphi(z)+\frac{b_2}{\alpha} z\varphi(z)',
\end{cases}
\end{equation}
where  $a_1,$ $a_2,$ $b_1$ and $b_2$ are constants. When $g(x)=\frac{f(x)'}{2}$, the reduced equations (\ref{eq:BU1ODE}) and (\ref{eq:BU2ODE}) have the following general form:
\begin{equation}\label{e16}
\frac{d^\alpha}{d z^\alpha}\varphi(z) = a\varphi(z)+\frac{b}{\alpha}z \varphi(z)',\quad
\end{equation}
where $a$ and $b$ are constants.
Thus, the problem of finding invariant solutions of (\ref{1}) is reduced to the problem of finding the solutions of (\ref{e15}) and (\ref{e16}). It should be noted that the solutions expressed in terms of the Mittag--Leffler function and the generalized Wright function are defined for $z\in\mathbb{R},$ whereas the solutions expressed in terms of the Fox H-function are defined for $z>0$.

We recall the following special functions for introducing exact invariant solutions of Eq.~(\ref{1}). The Fox H-function (e.g. in \cite{glockle,Kiryakova,Hnom}) is defined by means of the Mellin-Barnes type contour integral
\begin{equation}
\label{int1}
H_{p,q}^{m,l}\left[z\biggr\vert\begin{array}{c}
(a_i, \alpha_i)_{1,p}\\
(b_j, \beta_j)_{1,q}
\end{array}\right]=\frac{1}{2\pi i}\int_{L}\frac{\prod\limits_{j=1}^m\Gamma(b_j-\beta_js)\prod\limits_{i=1}^l\Gamma(1-a_i+\alpha_is)}{\prod\limits_{i=l+1}^p\Gamma(a_i-\alpha_is)\prod\limits_{j=m+1}^q\Gamma(1-b_j+\beta_js)}z^{s}d s,
\end{equation}
for $z\in\mathbb{C}\setminus\{0\},$ where $m,l,p,q\in\mathbb{N}_0=\{0,1,2,\ldots\}$, $(m,l)\neq (0,0),$ $\alpha_i, \beta_j\in\mathbb{R}_+$, $a_i,b_j\in\mathbb{R}$ ($i=1,\ldots, p;j=1,\ldots,q$). Here, $L$ is a suitable contour from $\gamma-i\infty$ to $\gamma+i\infty$, where $\gamma$ is a real number. The integral in (\ref{int1}) converges if the following conditions are met 
\begin{equation*}
\rho=\sum_{i=1}^{l}\alpha_i-\sum_{i=l+1}^{p}\alpha_i+\sum_{j=1}^{m}\beta_j-\sum_{j=m+1}^{q}\beta_j>0\quad\mbox{and}\quad |\arg z|<\frac{\pi\rho}{2}. 
\end{equation*}
The Fox H-function vanishes for large $z$ because 
\begin{equation*}
\label{eq8}
H_{p,q}^{m,0}[z]\approx O\left(\exp\left(-\nu z^{\frac{1}{\nu}}\mu^{\frac{1}{\nu}}\right)z^{\frac{2\delta+1}{2\nu}}\right),
\end{equation*}
where $\mu=\prod\limits_{i=1}^p\alpha_i^{\alpha_i}\prod\limits_{j=1}^q\beta_j^{-\beta_j},$ $\delta=\sum\limits_{j=1}^q b_j-\sum\limits_{i=1}^p a_i+\frac{p-q}{2}$  and
$\nu=\sum\limits_{j=1}^{q}\beta_j-\sum\limits_{i=1}^{p}\alpha_i>0$.

The generalized Wright function (e.g. in \cite{Luchko1,kilb2005}) is defined as
\begin{eqnarray}
\label{ub3}
{}_p\Psi_q\left[z\left|\begin{array}{c}
(a_i,\alpha_i)_{1,p}\\
(b_j,\beta_j)_{1,q}
\end{array}\right.\right] & = & \sum_{k=0}^\infty\frac{\prod\limits_{i=1}^p\Gamma(a_i+\alpha_i k)}{\prod\limits_{j=1}^q\Gamma(b_j+\beta_j k)}\frac{z^k}{k!},
\end{eqnarray}
for $z\in\mathbb{C},$ $p,q\in\mathbb{N}_0,$ $a_i,b_j\in\mathbb{C}$ and $\alpha_i, \beta_j\in\mathbb{R}\setminus\{0\}$ ($i=1,\ldots,p; j=1,\ldots,q$). 

If $\Delta=\sum\limits_{j=1}^q \beta_j-\sum\limits_{i=1}^p \alpha_i>-1$ or $\Delta=-1,$ then the series in (\ref{ub3}) is absolutely convergent for $z\in \mathbb{C}$ or $|z|<\prod\limits_{i=1}^p|\alpha_i|^{-\alpha_i}\prod\limits_{j=1}^q|\beta_j|^{\beta_j}$, respectively. 

Lastly, the Mittag-Leffler and Wright functions can be expressed in terms of the generalized Wright functions, respectively, as
\begin{equation*}
E_{\alpha,\beta}(z) = {}_1\Psi_1\left[z\left|\begin{array}{c}
(1,1)\\
(\beta,\alpha)
\end{array}\right.
\right]\quad\mbox{and}\quad 
\Psi\left(z;\alpha,\beta\right) ={}_0\Psi_{1}\left[z\left|\begin{array}{c}
-\\
(\beta,\alpha)
\end{array}\right.\right]. 
\end{equation*}

The solutions of (\ref{e15}) and (\ref{e16}) have been studied in detail in 
\cite{ourFCAA,ourJMP}. In particular, we obtain the following results based on Propositions 3.1 and 3.2 of \cite{ourFCAA} and Lemmas 2.1 and 2.2 of \cite{ourJMP}.

\textbf{Case 1. Solutions of (\ref{1}) with any $f(x)$ and $g(x)=\frac{\lambda_2\sqrt{f(x)}}{\omega_{\lambda_1}(x)}+\frac{f(x)'}{2}$ (here $\omega_{\lambda_1}(x)=\int_\beta^x \frac{dr}{\sqrt{f(r)}} +\lambda_1$, $\lambda_1,\lambda_2\in \mathbb{R},~ \lambda_2\ne 0$)}

In this case, invariant solutions of (\ref{eq:TEcaseA2}) are obtained as follows:
\begin{itemize}
\item[1.] For $0 <\alpha<1$, if we apply the third assertion of Proposition~3.2 of
\cite{ourFCAA} 
with $a_1 = a$, $a_2 = a+\lambda_2$ and
$b_1 = b_2 =-1$ in (\ref{eq:RS_thm_ii}), we obtain the following solutions through (\ref{eq:ST_thm_ii}):
\begin{eqnarray}\label{sol1}
u(x,t)&=&\frac{c}{\sqrt{f(x)}}\omega_{\lambda_1}(x)^{a}
\varphi\left(\omega_{\lambda_1}(x)^{-\frac{1}{\alpha}} t\right),\nonumber\\
v(x,t)&=& -c\omega_{\lambda_1}(x)^a
\psi\left(\omega_{\lambda_1}(x)^{-\frac{1}{\alpha}} t\right),
\end{eqnarray}
where $c$ is a constant and
\begin{eqnarray*}
\varphi(z) & = &  H_{1,2}^{2,0}\left[\frac{z^{-2\alpha}}{4}\biggr\vert\begin{array}{c}
(1,2\alpha)\\
\left(\frac{1}{2}-\frac{a}{2},1\right), \left(-\frac{a+\lambda_2}{2},1\right)
\end{array}\right],\\
\psi(z) & = &  H_{1,2}^{2,0}\left[\frac{z^{-2\alpha}}{4}
\biggr\vert\begin{array}{c}
(1,2\alpha)\\
\left(-\frac{a}{2},1\right), \left(\frac{1}{2}-\frac{a+\lambda_2}{2},1\right)
\end{array}\right].
\end{eqnarray*}
\item[2.] For $\alpha\ge 1,$ applying the third assertion of Proposition~3.1 of \cite{ourFCAA} 
with the same parameters as in item 1 in (\ref{eq:RS_thm_ii}), we obtain the following solutions through (\ref{eq:ST_thm_ii}):
\begin{eqnarray}
u(x,t)&=&\frac{1}{\sqrt{f(x)}}\omega_{\lambda_1}(x)^{a}
\varphi\left(\omega_{\lambda_1}(x)^{-\frac{1}{\alpha}} t\right),\\
v(x,t)&=& \omega_{\lambda_1}(x)^a
\psi\left(\omega_{\lambda_1}(x)^{-\frac{1}{\alpha}} t\right),
\end{eqnarray}
where  
\begin{eqnarray*}
\varphi(z)&=&\sum_{k=1}^{n}c_{k,1}z^{\alpha-k} {}_3\Psi_1\left[4 z^{2\alpha}\biggr\vert
\begin{array}{c}
\left(1-\frac
{a}{2}-\frac{k}{2\alpha},1\right),\left(\frac{1}{2}-\frac{a+\lambda_2}{2}-\frac{k}{2\alpha},1\right),(1,1)\\
(1+\alpha-k,2\alpha)
\end{array}\right]\\
& & -2\sum_{k=1}^n c_{k,2}z^{2\alpha-k} {}_3\Psi_1\left[4 z^{2\alpha}\biggr\vert\begin{array}{c}\left(\frac{3}{2}-\frac
{a}{2}-\frac{k}{2\alpha},1\right),\left(1-\frac{a+\lambda_2}{2}-\frac{k}{2\alpha},1\right),(1,1)\\
(1+2\alpha-k,2\alpha)
\end{array}\right],\\
\psi(z)&=&-2\sum_{k=1}^n c_{k,1}z^{2\alpha-k} {}_3\Psi_1\left[4 z^{2\alpha}\biggr\vert\begin{array}{c}
\left(1-\frac
{a}{2}-\frac{k}{2\alpha},1\right),\left(\frac{3}{2}-\frac{a+\lambda_2}{2}-\frac{k}{2\alpha},1\right),(1,1)\\
(1+2\alpha-k,2\alpha)
\end{array}\right]\\
& & +\sum_{k=1}^n c_{k,2}z^{\alpha-k} {}_3\Psi_1\left[4 z^{2\alpha}\biggr\vert\begin{array}{c}
\left(\frac{1}{2}-\frac
{a}{2}-\frac{k}{2\alpha},1\right),\left(1-\frac{a+\lambda_2}{2}-\frac{k}{2\alpha},1\right),(1,1)\\
(1+\alpha-k,2\alpha)
\end{array}\right],
\end{eqnarray*}
where $n\in\mathbb{N}$ satisfies $0\leq n-1<\alpha\le n$  and $c_{k,1},$ $c_{k,2}$ $(k=1,\dots,n)$ are constants.
\end{itemize}

\textbf{Case 2. Solutions of (\ref{1}) with any $f(x)$ and $g(x)=\lambda_2 \sqrt{f(x)}+\frac{f(x)'}{2}$ (here $\lambda_2\in \mathbb{R},~ \lambda_2\ne 0$)}

In this case, using the first assertion of Lemma~3.1 of \cite{ourJMP} 
with $a_1 = a$ and $a_2 = a + \lambda_2$ in (\ref{eq:RS_thm_iii}), we
obtain the following solutions  of (\ref{eq:TEcaseA3}) through (\ref{eq:ST_thm_iii}):
\begin{eqnarray}\label{sol2}
u(x,t)&=&\frac{1}{\sqrt{f(x)}}\exp\left(a\omega_0(x)\right)\varphi(t),\nonumber\\
v(x,t)&=&\exp\left(a\omega_0(x)\right)\psi(t),
\end{eqnarray}
where 
\begin{align*}
\omega_{0}(x)&=\int^x_{\beta} \frac{1}{\sqrt{f(r)}}dr,\\
\varphi(t) & =  \sum_{k=1}^n c_{k,1}t^{\alpha-k} E_{
2\alpha,1+\alpha-k}\left(a(a+\lambda_2)t^{2\alpha}\right)
+a\sum_{k=1}^nc_{k,2}t^{2\alpha-k} E_{2\alpha,1+2\alpha-k}\left(a(a+\lambda_2)t^{2\alpha}\right),\\
\psi(t) & =  (a+\lambda_2)\sum_{k=1}^n c_{k,1}t^{2\alpha-k} E_{
2\alpha,1+2\alpha-k}\left(a(a+\lambda_2)t^{2\alpha}\right)+\sum_{k=1}^nc_{k,2}t^{\alpha-k} E_{
2\alpha,1+\alpha-k}
\left(a(a+\lambda_2)t^{2\alpha}\right),
\end{align*}
where $n\in\mathbb{N}$ satisfies $0\leq n-1<\alpha\le n,$  and $c_{k,1},$ $c_{k,2}$ $(k=1,\dots,n)$ are constants.

This is a compelling instance because (\ref{1}) can be reduced to separable equations for any function $f(x)$. When $g(x)=0$, however, only $f(x)=x^2$ leads to a separable equation 
\cite{ourJMP}. 
Nevertheless, if $f(x)=x^{2m}, ~\forall m>0$ and $g(x)=x^m(\lambda_2+mx^{m-1}),$ then the system (\ref{1}) is reduced into separable form.

\textbf{Case 3. Solutions of (\ref{1}) with any $f(x)$ and $g(x)=\frac{f(x)'}{2}$.}
Unlike the previously cases of $f(x)$ and $g(x)$, the invariant solutions of (\ref{1}) with any $f(x)$ and $g(x)=\frac{f(x)'}{2}$ are expressed
as the sum of the solutions of the two individual FODEs. First, invariant solutions of (\ref{usys1})
corresponding to $W_4$ are given as follows:
\begin{enumerate}
\item If $0<\alpha<1$, using the second assertion of Lemma~3.2 of \cite{ourJMP} 
with $a=a_2-a_1$ and $b=1$ in the second equation of the reduced system in (\ref{eq:BU1ODE}), we obtain the following invariant solutions corresponding to $W_4$
\begin{eqnarray}\label{sol3}
u(x,t)&=&\frac{ct^{(a_1-a_2)\alpha}}{\sqrt{f(x)}}\Psi\left(-\frac{\omega_0(x)}{t^{\alpha}};-\alpha,1+(a_1-a_2)\alpha\right)\nonumber,\\
v(x,t)&=&-ct^{(a_1-a_2)\alpha}\Psi\left(-\frac{\omega_0(x)}{t^{\alpha}};-\alpha,1+(a_1-a_2)\alpha\right).
\end{eqnarray}

\item If $\alpha \ge 1$, using the third assertion of Lemma~3.2 
\cite{ourFCAA}  
with $a = a_1 + a_2$, $b = -1$ for the first equation
and $a = a_2 - a_1$, $b = 1$ for the second equation of the reduced system in (\ref{eq:BU1ODE}), we obtain the following invariant solutions corresponding to $W_4$
\begin{eqnarray}\label{c3}
u(x,t) & = & \frac{1}{\sqrt{f(x)}}\omega_0(x)^{a_1+a_2-1}\sum\limits_{k=1}^nc_{k,1}\omega_0(x)^{\frac{k}{\alpha}}t^{\alpha-k}\varphi_k\left(\omega_0(x)^{-\frac{1}{\alpha}}t\right)\\
&&+\frac{1}{\sqrt{f(x)}}\omega_0(x)^{a_1-a_2-1}\sum\limits_{k=1}^nc_{k,2}\omega_0(x)^{\frac{k}{\alpha}}t^{\alpha-k}\psi_k\left(\omega_0(x)^{-\frac{1}{\alpha}}t\right),\nonumber\\
v(x,t) & = & \omega_0(x)^{a_1+a_2-1}\sum\limits_{k=1}^nc_{k,1}\omega_0(x)^{\frac{k}{\alpha}}t^{\alpha-k}\varphi_k\left(\omega_0(x)^{-\frac{1}{\alpha}}t\right)\\
&&-\omega_0(x)^{a_1-a_2-1}\sum\limits_{k=1}^nc_{k,2}\omega_0(x)^{\frac{k}{\alpha}}t^{\alpha-k}\psi_k\left(\omega_0(x)^{-\frac{1}{\alpha}}t\right),
\end{eqnarray}
where $n\in\mathbb{N}$ satisfies $0\leq n-1<\alpha\le n,$  and $c_{k,1},$ $c_{k,2}$ $(k=1,\dots,n)$ are constants,
\begin{eqnarray*}
\omega_{0}(x)&=&\int^x_{\beta} \frac{1}{\sqrt{f(r)}}dr,\\
\varphi_k(z) &=& {}_2\Psi_1\left[-z^{\alpha}\biggr\vert\begin{array}{c}
\left(-a_1-a_2-\frac{k}{\alpha}+1,1\right),(1,1)\\
(1+\alpha-k,\alpha)
\end{array}\right],\\
\psi_k(z) &=&  {}_2\Psi_1\left[z^{\alpha}\biggr\vert\begin{array}{c}
\left(-a_1+a_2-\frac{k}{\alpha}+1,1\right),(1,1)
\\
(1+\alpha-k,\alpha)
\end{array}\right].
\end{eqnarray*}
\end{enumerate}

Since $b=0$ in both equations of the reduced system in (\ref{eq:BU2ODE}), we obtain the following invariant solutions corresponding to $W_5$ using the first assertion of Lemma~3.2 of \cite{ourJMP}:
\begin{eqnarray*}\label{sol4}
u(x,t) & = & \frac{1}{\sqrt{f(x)}}\exp{\left((a_1+a_2)\omega_0(x)\right)}\sum\limits_{k=1}^nc_{k,1}t^{\alpha-k}\varphi_k(t)\\
&&~~+\frac{1}{\sqrt{f(x)}}\exp{\left((a_1-a_2)\omega_0(x)\right)}\sum\limits_{k=1}^nc_{k,2}t^{\alpha-k}\psi_k(t),\nonumber\\
v(x,t) & = & \exp{\left((a_1+a_2)\omega_0(x)\right)}\sum\limits_{k=1}^nc_{k,1}t^{\alpha-k}\varphi_k(t)-\exp{\left((a_1-a_2)\omega_0(x)\right)}\sum\limits_{k=1}^nc_{k,2}t^{\alpha-k}\psi_k(t),
\end{eqnarray*}
where $\omega_{0}(x)=\int^x_{\beta} \frac{1}{\sqrt{f(r)}}dr$, $\varphi_k(t)=E_{\alpha,1+\alpha-k}((a_1+a_2)t^{\alpha}), \quad  \psi_k(t) = E_{\alpha,1+\alpha-k}((a_2-a_1)t^{\alpha})$
and $(a_1,a_2)\in\{(\pm 1,a), (0,\pm 1), (0,0)\vert a\in\mathbb{R}\}$.
Because this case admits the same symmetries as the previous two cases, it can be reduced to both separable and inseparable forms. 

\section{Conclusion}

In this study, we established a complete Lie symmetry classification for a class of time-fractional telegraph systems with spatially varying coefficients. These systems represent fractional generalizations of classical telegraph equations and provide important mathematical models for transport phenomena involving memory and nonlocal effects. Such fractional telegraph models naturally arise in industrial and physical applications, including heat transport in materials with thermal memory, wave propagation in viscoelastic media, and charge transport in spatially heterogeneous semiconductor devices.

Our analysis shows that the symmetry structure of the system depends fundamentally on the relationship between the transport coefficient $f(x)$ and the potential function $g(x)$. We determined all functional forms of the coefficient functions that admit symmetry extensions and classified the corresponding Lie symmetry algebras into three distinct symmetry classes. This result provides a complete group classification for the class of time fractional telegraph systems under consideration and significantly extends existing symmetry classification results for fractional evolution and diffusion equations, which were previously limited to specific coefficient forms or special cases.

Using the optimal system methodology, we systematically reduced the governing fractional telegraph equations to fractional ordinary differential equations. This reduction enabled the explicit construction of invariant solutions in closed analytical form. The obtained solutions are expressed in terms of Mittag–Leffler functions, generalized Wright functions, and Fox $H$-functions, which naturally characterize fractional transport processes with memory. These analytical solutions provide valuable insights into the mathematical structure and physical behavior of fractional telegraph systems and serve as important benchmarks for validating the numerical methods used in fractional transport modeling.

The results obtained in this study extend the classical symmetry analysis of telegraph equations to the fractional setting with general spatially varying coefficients. In particular, the complete symmetry classification presented here broadens the class of fractional telegraph systems for which invariant solutions can be systematically constructed and analyzed. This contributes to the mathematical foundation of fractional transport theory and enhances the applicability of symmetry methods in modeling transport phenomena in complex industrial systems.

To the best of our knowledge, this is the first work that provides a complete Lie symmetry classification and explicit invariant solutions for time-fractional telegraph systems with general spatially varying coefficients.
Future work may focus on extending the present analysis to nonlinear fractional telegraph equations, multidimensional fractional transport systems, and models involving alternative fractional operators such as Caputo and generalized fractional derivatives. These extensions would further strengthen the role of symmetry analysis in the study of fractional transport processes and support the development of accurate mathematical models for industrial systems with memory and spatial heterogeneity.

\subsection*{Acknowledgments}
 This work was supported by the National University of Mongolia (Grant No.P2020-3966)

\bibliographystyle{unsrt} 
\bibliography{ref}

@book{BlumanKumei1989,
  author    = {Bluman, George W. and Kumei, Sukeyuki},
  title     = {Symmetries and Differential Equations},
  series    = {Applied Mathematical Sciences},
  volume    = {81},
  publisher = {Springer-Verlag},
  address   = {New York},
  year      = {1989},
  doi       = {10.1007/978-1-4612-4478-3}
}

@article{HuangShen2015,
  author  = {Huang, Qing and Shen, Shoufeng},
  title   = {Lie symmetries and group classification of a class of time-fractional evolution systems},
  journal = {Journal of Mathematical Physics},
  volume  = {56},
  number  = {12},
  pages   = {123504},
  year    = {2015},
  doi     = {10.1063/1.4937399}
}

@article{ourJMP,
  author  = {Dorjgotov, Khongorzul and Ochiai, Hiroyuki and Zunderiya, Uuganbayar},
  title   = {Exact solutions to a class of time-fractional evolution systems with variable coefficients},
  journal = {Journal of Mathematical Physics},
  volume  = {59},
  number  = {8},
  pages   = {081504},
  year    = {2018},
  doi     = {10.1063/1.5026462}
}

@article{ourNL,
  author  = {Dorjgotov, Khongorzul and Ochiai, Hiroyuki and Zunderiya, Uuganbayar},
  title   = {Lie symmetry analysis of a class of time-fractional nonlinear evolution systems},
  journal = {Applied Mathematics and Computation},
  volume  = {329},
  pages   = {105--117},
  year    = {2018},
  doi     = {10.1016/j.amc.2018.02.028}
}

@book{Podlubny1999,
  author    = {Podlubny, Igor},
  title     = {Fractional Differential Equations},
  series    = {Mathematics in Science and Engineering},
  volume    = {198},
  publisher = {Academic Press},
  address   = {San Diego},
  year      = {1999}
}

@article{ourFCAA,
  author  = {Dorjgotov, Khongorzul and Ochiai, Hiroyuki and Zunderiya, Uuganbayar},
  title   = {On solutions of linear fractional differential equations and systems thereof},
  journal = {Fractional Calculus and Applied Analysis},
  volume  = {22},
  number  = {2},
  pages   = {479--494},
  year    = {2019},
  doi     = {10.1515/fca-2019-0028}
}

@book{Olver1993,
  author    = {Olver, Peter J.},
  title     = {Applications of Lie Groups to Differential Equations},
  edition   = {2nd},
  series    = {Graduate Texts in Mathematics},
  volume    = {107},
  publisher = {Springer-Verlag},
  address   = {New York},
  year      = {1993},
  doi       = {10.1007/978-1-4684-0533-0}
}

@article{GazizovKasatkinLukashchuk2007,
  author  = {Gazizov, Rafail K. and Kasatkin, Alexander A. and Lukashchuk, Stanislav Yu.},
  title   = {Continuous transformation groups of fractional differential equations},
  journal = {Symmetry, Integrability and Geometry: Methods and Applications (SIGMA)},
  volume  = {3},
  pages   = {002},
  year    = {2007},
  doi     = {10.3842/SIGMA.2007.002}
}

@book{MillerRoss1993,
  author    = {Miller, Kenneth S. and Ross, Bertram},
  title     = {An Introduction to the Fractional Calculus and Fractional Differential Equations},
  publisher = {John Wiley \& Sons},
  address   = {New York},
  year      = {1993}
}

@article{Gurefe2023,
  author  = {Gurefe, Yusuf},
  title   = {The fractional derivative with respect to another function and its application to Lie symmetry analysis},
  journal = {Chaos, Solitons \& Fractals},
  volume  = {170},
  pages   = {113335},
  year    = {2023},
  doi     = {10.1016/j.chaos.2023.113335}
}

@article{PateraWinternitz1977,
  author  = {Patera, J. and Winternitz, P.},
  title   = {Subalgebras of real three- and four-dimensional Lie algebras},
  journal = {Journal of Mathematical Physics},
  volume  = {18},
  number  = {7},
  pages   = {1449--1455},
  year    = {1977},
  doi     = {10.1063/1.523453}
}

@article{Baleanu2022,
  author  = {Baleanu, Dumitru and Sajjad, Shama and Jamshed, Waqas and Safdar, M.},
  title   = {Lie Symmetry Analysis of Fractional Differential Equations},
  journal = {Fractals},
  volume  = {30},
  number  = {04},
  pages   = {2250075},
  year    = {2022},
  doi     = {10.1142/S0218348X2250075X}
}

@article{Wang2023,
  author  = {Wang, Ying-Hui and Deng, Hou-Dao and Cheng, Wei-Guo},
  title   = {Lie symmetry analysis, optimal system and exact solutions of the (2+1)-dimensional time-fractional Kadomtsev-Petviashvili equation},
  journal = {Fractional Calculus and Applied Analysis},
  volume  = {26},
  number  = {1},
  pages   = {240--264},
  year    = {2023},
  doi     = {10.1007/s13540-022-00108-w}
}

@book{Mainardi2010,
  author    = {Mainardi, Francesco},
  title     = {Fractional Calculus and Waves in Linear Viscoelasticity: An Introduction to Mathematical Models},
  publisher = {Imperial College Press},
  address   = {London},
  year      = {2010},
  doi       = {10.1142/p614}
}

@article{OrsingherBeghin2004,
  author  = {Orsingher, Enzo and Beghin, Luisa},
  title   = {Time-fractional telegraph equations and telegraph processes with Brownian time},
  journal = {Probability Theory and Related Fields},
  volume  = {128},
  number  = {1},
  pages   = {141--160},
  year    = {2004},
  doi     = {10.1007/s00440-003-0308-4}
}

@article{MetzlerKlafter2000,
  author  = {Metzler, Ralf and Klafter, Joseph},
  title   = {The random walk's guide to anomalous diffusion: a fractional dynamics approach},
  journal = {Physics Reports},
  volume  = {339},
  number  = {1},
  pages   = {1--77},
  year    = {2000},
  doi     = {10.1016/S0370-1573(00)00070-3}
}

@article{HenryWearne2000,
  author  = {Henry, Bruce I. and Wearne, Stephen L.},
  title   = {Fractional reaction--diffusion},
  journal = {Physica A: Statistical Mechanics and its Applications},
  volume  = {276},
  number  = {3--4},
  pages   = {448--455},
  year    = {2000},
  doi     = {10.1016/S0378-4371(99)00469-9}
}

@article{ScherMontroll1975,
  author  = {Scher, Harvey and Montroll, Elliott W.},
  title   = {Anomalous transit-time dispersion in amorphous solids},
  journal = {Physical Review B},
  volume  = {12},
  number  = {6},
  pages   = {2455--2477},
  year    = {1975},
  doi     = {10.1103/PhysRevB.12.2455}
}

@article{Kac1974,
  author  = {Kac, Mark},
  title   = {A stochastic model related to the telegrapher's equation},
  journal = {Rocky Mountain Journal of Mathematics},
  volume  = {4},
  number  = {3},
  pages   = {497--509},
  year    = {1974}
}

@article{BlumanTemuerchaoluSahadevan2005,
  author  = {Bluman, George and Temuerchaolu and Sahadevan, R.},
  title   = {Local and nonlocal symmetries for nonlinear telegraph equations},
  journal = {Journal of Mathematical Physics},
  volume  = {46},
  number  = {2},
  pages   = {023505},
  year    = {2005},
  doi     = {10.1063/1.1835507}
}

@article{BlumanTemuerchaolu2005JMAA,
  author  = {Bluman, George and Temuerchaolu},
  title   = {Conservation laws for nonlinear telegraph equations},
  journal = {Journal of Mathematical Analysis and Applications},
  volume  = {310},
  number  = {2},
  pages   = {459--476},
  year    = {2005},
  doi     = {10.1016/j.jmaa.2005.01.012}
}

@article{BlumanTemuerchaolu2005JMP,
  author  = {Bluman, George and Temuerchaolu},
  title   = {Comparing symmetries and conservation laws of nonlinear telegraph equations},
  journal = {Journal of Mathematical Physics},
  volume  = {46},
  number  = {7},
  pages   = {073513},
  year    = {2005},
  doi     = {10.1063/1.1946907}
}

@article{LukashchukMakunin2015,
  author  = {Lukashchuk, Stanislav Yu. and Makunin, A. V.},
  title   = {Group classification of nonlinear time-fractional diffusion equation with a source term},
  journal = {Applied Mathematics and Computation},
  volume  = {257},
  pages   = {335--343},
  year    = {2015},
  doi     = {10.1016/j.amc.2014.11.087}
}

@article{glockle,
  author       = {Gl{\"o}ckle, Walter G. and Nonnenmacher, Theo F.},
  title        = {Fox function representation of non-Debye relaxation processes},
  journal      = {Journal of Statistical Physics},
  volume       = {71},
  number       = {3-4},
  pages        = {741--757},
  year         = {1993},
  doi          = {10.1007/BF01058445}
}

@book{Hnom,
  author    = {Mathai, A. M. and Saxena, Ram Kishore and Haubold, Hans J.},
  title     = {The H-Function: Theory and Applications},
  publisher = {Springer-Verlag},
  address   = {New York},
  year      = {2010},
  doi       = {10.1007/978-1-4419-3002-6}
}

@article{Kiryakova,
  author  = {Kiryakova, Virginia},
  title   = {Fractional calculus operators of special functions? The result is well predictable!},
  journal = {Chaos, Solitons \& Fractals},
  volume  = {102},
  pages   = {2--15},
  year    = {2017},
  doi     = {10.1016/j.chaos.2017.05.006}
}

@article{Luchko1,
  author    = {Gorenflo, Rudolf and Luchko, Yuri and Mainardi, Francesco},
  title     = {Wright functions as scale-invariant solutions of the diffusion-wave equation},
  journal   = {Journal of Computational and Applied Mathematics},
  volume    = {118},
  number    = {1-2},
  year      = {2000},
  pages     = {175--191},
  doi       = {10.1016/S0377-0427(00)00288-0}
}

@article{kilb2005,
  author  = {Kilbas, Anatoly A.},
  title   = {Fractional calculus of the generalized Wright function},
  journal = {Fractional Calculus and Applied Analysis},
  volume  = {8},
  number  = {2},
  pages   = {113--126},
  year    = {2005}
}
\end{document}